\documentclass[letterpaper, 10 pt, conference]{ieeeconf}  % Comment this line out
\IEEEoverridecommandlockouts                              % This command is only
\usepackage{amsmath} % assumes amsmath package installed
\usepackage{amssymb}  % assumes amsmath package installed
\usepackage{amsfonts}
\usepackage{graphicx} 
\usepackage[usenames,dvipsnames]{color}

\usepackage[]{ifthen}
\newtheorem{theorem}{Theorem}
\newtheorem{lemma}{Lemma}
\newtheorem{corollary}{Corollary}
\newtheorem{remark}{Remark}
\newboolean{showcomments}
\setboolean{showcomments}{true}
\newcommand{\enrique}[1]{\ifthenelse{\boolean{showcomments}}
{\textcolor{Red}{(Enrique says: #1)}}{}}
\newcommand{\addcite}[0]{\ifthenelse{\boolean{showcomments}}
{\textcolor{Purple}{(add cite(s))}}{}}
\newcommand{\addcites}[0]{\ifthenelse{\boolean{showcomments}}
{\textcolor{Purple}{(add cite(s))}}{}}

\newcommand{\abs}[1]{\ensuremath{\left\vert#1\right\vert}}

\newcommand{\diag}[1]{\ensuremath{\mathrm{diag}\left(#1\right)}}

\newcommand{\funof}[1]{\ensuremath{\left(#1\right)}}

\newcommand{\jw}{\ensuremath{\left(j\omega\right)}}

\newcommand{\R}{\ensuremath{\mathbb{R}}}

\newcommand{\s}{\ensuremath{\left(s\right)}}

\DeclareMathOperator{\tr}{tr}

\usepackage{enumerate}

\title{
Performance tradeoffs of dynamically controlled grid-connected  inverters in low inertia power systems
}

\author{Yan Jiang, Richard Pates, and Enrique Mallada% <-this % stops a space
\thanks{This work was supported by NSF CPS grant CNS 1544771, Johns Hopkins E$^2$SHI Seed Grant, and Johns Hopkins WSE startup funds.
}% <-this % stops a space
\thanks{Yan Jiang and Enrique Mallada are with the Department of Electrical and Computer Engineering, Johns Hopkins University, Baltimore, MD 21218 USA (e-mails: \{yjiang,mallada\}@jhu.edu).
}
\thanks{Richard Pates is with the Department of Automatic Control at Lund University, Lund, Sweden (e-mail: richard.pates@control.lth.se)
}
}

\renewcommand{\baselinestretch}{.985}

\begin{document}

\maketitle
\thispagestyle{empty}
\pagestyle{empty}

%%%%%%%%%%%%%%%%%%%%%%%%%%%%%%%%%%%%%%%%%%%%%%%%%%%%%%%%%%%%%%%%%%%%%%%%%%%%%%%%
\begin{abstract}
Implementing frequency response using grid-connected inverters is one of the popular proposed alternatives to mitigate the dynamic degradation experienced in low inertia power systems. However, such solution faces several challenges as inverters do not intrinsically possess the natural response to power fluctuations that synchronous generators have. Thus, to synthetically generate this response, inverters need to take frequency measurements, which are usually noisy, and subsequently make changes in the output power, which are therefore delayed. This paper explores the system-wide performance tradeoffs that arise when measurement noise, power disturbances, and delayed actions are considered in the design of dynamic controllers for grid-connected inverters. 
Using a recently proposed dynamic droop (iDroop) control for grid-connected inverters, which is inspired by classical first order lead-lag compensation, we show that the sets of parameters that result in highest noise attenuation, power disturbance mitigation, and delay robustness do not necessarily have a common intersection. In particular, lead compensation is desired in systems where power disturbances are the predominant source of degradation, while lag compensation is a better alternative when the system is dominated by delays or frequency noise. Our analysis further shows that iDroop can outperform the standard droop alternative in both joint noise and disturbance mitigation, and delay robustness.
\end{abstract}

%%%%%%%%%%%%%%%%%%%%%%%%%%%%%%%%%%%%%%%%%%%%%%%%%%%%%%%%%%%%%%%%%%%%%%%%%%%%%%%%
%INTRODUCTION
%%%%%%%%%%%%%%%%%%%%%%%%%%%%%%%%%%%%%%%%%%%%%%%%%%%%%%%%%%%%%%%%%%%%%%%%%%%%%%%%
%!TEX root = main.tex
\section{Introduction}\label{sec:intro}
%The power grid is transforming.
The composition of the electric gird is in state of flux~\cite{Milligan:2015ju}. 
Motivated by the need of reducing carbon emissions, conventional synchronous generators, with relatively large inertia, are being replaced with renewable energy sources with little (wind) or no inertia (solar) at all~\cite{Winter:2015dy}. 
Alongside, neither the remaining generators, nor the demand, are compensating this loss. On the synchronous generator side, there are no incentives to provide additional frequency response beyond their natural one~\cite{EasternInterconnect:2013wq}. On the demand side, the steady increase of power electronics is gradually diminishing the load sensitivity to frequency variations~\cite{WoodWollenberg1996}.
As a result, rapid frequency fluctuations are becoming a major source of concern for several grid operators~\cite{Boemer:2010wa,Kirby:2005uy}. 
Besides increasing the risk of frequency instabilities, this dynamic degradation also places limits on the total amount of renewable generation that can be sustained by the grid. Ireland, for instance, is already resourcing to wind curtailment --whenever wind becomes larger than $50\%$ of existing demand-- in order to preserve the grid stability.

% Existing solutions aim to mimic old ones for no apparent reason.
Among the several efforts under way to mitigate this dynamic degradation, one prominent alternative is to implement frequency response using (electronically coupled) inverter-based generation~\cite{Spolar:2012vj}.
For example, in the US, the Federal Energy Regulatory Commission (FERC) has recently issued a notice of intent of rule making~\cite{FederalEnergyRegulatoryCommission:2016wz} that mandates frequency response by renewable generation.
The goal is to use inverter-based generators to mimic synchronous generator behavior, or in other words, to implement virtual inertia~\cite{Driesen:ft}.
However, while implementing virtual inertia can mitigate this degradation, it is unclear whether that particular choice of control is the most suitable for it.
On the one hand, unlike generator dynamics that set the grid frequency, virtual inertia controllers estimate the grid frequency and its derivative using noisy and delayed measurements.
On the other hand, inverter-based control can be significantly faster than conventional generators. Thus using inverters to mimic generators does not take advantage of their full potential.

Recently, a novel dynamic droop control (iDroop)~\cite{m2016cdc} has been proposed as an alternative to virtual inertia. iDroop uses first order lead-lag compensation --inspired by scalable control laws in data networks~\cite{Paganini:gz}--  and seeks to exploit the added flexibility present in inverters. 
Unlike virtual inertia that is sensitive to noisy measurements (it has unbounded $H_2$ norm~\cite{m2016cdc}), iDroop  experimentally improves the dynamic performance without the undesired unbounded noise amplification. 
In this paper we provide a theoretical foundation to such experimental findings. 
More precisely, for networks with homogeneous parameters, we analytically compute the dynamic performance ($H_2$ norm) of the control law proposed in~\cite{m2016cdc} in the presence of both frequency measurements and power disturbances (Theorem \ref{th:h2-idroop}). We show that iDroop not only is able to mitigate the noise amplification that virtual inertia introduces, but it can also outperform the standard droop control (Theorem \ref{th:h2-improves}). Furthermore, %using a novel decentralized stability analysis for power systems~\cite{pm2017ifac-wc} 
we analyze the robust stability of iDroop in the presence of delay and show that it can also outperform droop control (Theorem \ref{thm:delay}).

The analysis also unveils several intrinsic performance tradeoffs between power disturbances, measurement noise and delays, and how the lead-lag structure of iDroop is instrumental on the performance improvements. In particular, when the system is dominated by power disturbances, lead compensation provides the best performance. However, when the system is dominated by frequency noise, lag compensation is desired. Interestingly, the latter (lag compensation) matches the requirements for improving delay robustness with iDroop. However, achieving joint disturbance attenuation and delay robustness can be more challenging.

\section{Preliminaries}\label{sec:prelim}

% In this section we introduce the power network model used in our analysis as well as several tools used in our analysis.

\subsection{Power Network Model}\label{ssec:model}

We use the power network model used in \cite{pm2017ifac-wc}. We consider a network of $n$ buses denoted by $i\in V := \{1,\dots,n\}$. The power system dynamics is modeled as the feedback interconnection of bus dynamics $P := \diag{p_i, i \in V}$ and network dynamics $N$ as shown in Figure \ref{fig:1}.

\begin{figure}
\centering
\includegraphics[height=6cm,width=.85\columnwidth]{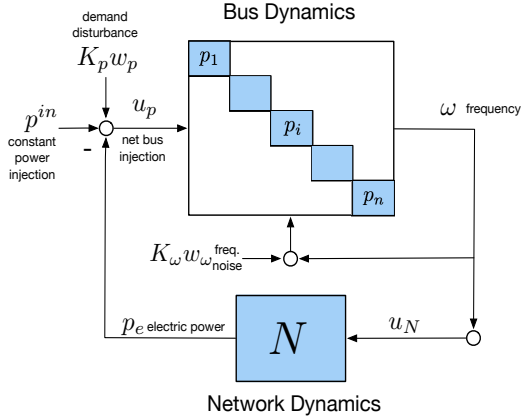}
\caption{Power Network Model}\label{fig:1}
\end{figure}

\subsubsection*{Bus Dynamics}
Each subsystem $p_i$ describes the $i$th bus dynamics where the input is the net bus power injection imbalance $u_{p,i}$ and the output is the frequency deviation from the nominal value $\omega_i$.
The bus dynamics of this paper are described in Figure \ref{fig:2} where $g_i$ represents the generator dynamics and $c_i$ the inverters dynamics. We assume that both dynamics are linear time invariant and thus 
\begin{equation}\label{eq:pi-1}
\hat \omega_i(s) = p_i(s) \hat u_{p,i}(s),%\;\text{ and }\; p_i(s)= \frac{g_i(s)}{1+c_i(s)g_i(s)},
\end{equation}
where $\hat \omega_i(s)$ and $\hat u_{p,i}(s)$ are the Laplace transform of $\omega_i(t)$ and $u_{p,i}(t)$, respectively.

The generator dynamics map the power injection $x_{i}(t)+u_{p,i}(t)$ to the bus frequency $\omega_i(t)$, and are represented in Laplace domain by
\begin{equation}\label{eq:gi}
\hat \omega_i(s) = g_i(s)\left(\hat x_i(s) + \hat u_{p,i}(s)\right).
\end{equation}
We use the  swing dynamics to model $g_i$, i.e., 
\begin{equation}\label{eq:swing}
	g_i(s)  = \frac{1}{M_is + D_i}
\end{equation}
where $M_i$ is the aggregate generator's inertia and $D_i$ is the aggregate generator's droop and frequency dependent load coefficient.

The inverter dynamics are modeled as
\begin{equation}\label{eq:ci}%\label{eq:02}
	\hat{x}_i(s) = - c_i(s)\hat{\omega}_i(s) 
\end{equation}
where $\hat{x}_i(s)$ denotes the Laplace transform of $x_i(t)$ (the power injected by the inverter), and $c_i(s)$ represents the control law. 
Equation \eqref{eq:ci} assumes that inverters operate in frequency synchronized (grid-connected) mode~\cite{Ciobotaru:2008kp,Blaabjerg:2006cea} where each inverter measures the local grid frequency $\omega_i(t)$ and statically sets the voltage phase of the inverter so that the output power is $x_i(t)$. This is a reasonable assumption as generator dynamics are significantly slower than inverters. 

% \begin{subequations} \label{eq:01}
%     \begin{equation}
%         M_i \dot{\omega_i} = - (R_{g,i}^{-1} + D_i)\omega_i + x_i + u_{p,i}
%     \end{equation}
%     \begin{equation}
%         y_{p,i} = \omega_i 
%     \end{equation} 
% \end{subequations} 
% where $\omega_i$ denotes the frequency deviation from nominal, $x_i$ the power injected by inverter-based generation at bus $i$, $u_{p,i}$ the exogenous real power injection to bus $i$, $y_{p,i}$ the output of subsystem $p_i$, $M_i$ the aggregate bus inertia, $D_i$ the aggregate damping  or frequency dependent load coefficient, and $R_{g,i}$ the droop coefficient.

We can use the control law $c_i(s)$ to model different algorithms that can be used to map $\omega_i(t)$ to $x_i(t)$. For example, 
$c_i(s) = R_{r,i}^{-1}$ and $c_i(s) = (\nu_is + R_{r,i}^{-1})$,
% \begin{equation}\label{eq:dc+vi}
% c_i(s) = -R_{r,i}^{-1}\quad\text{ or }\quad c_i(s) = -(\nu_is + R_{r,i}^{-1}),
% \end{equation}
represent the standard droop and virtual inertia controllers, respectively. Similarly, the iDroop controller defined in \cite{m2016cdc} is given by
\begin{equation}
        c_i(s) = \dfrac{\nu_i s + \delta_i R_{r,i}^{-1}}{s + \delta_i} \; \label{eq:05}
\end{equation} 
where $R_{r,i}^{-1}$ is the droop constant, and $\nu_i$ and $\delta_i$ are tunable parameters.

Combining (\ref{eq:gi}) and (\ref{eq:ci}) gives the input-output representation of the bus dynamics as
\begin{equation}
        p_i(s) = \dfrac{1}{M_i s + D_i}\left(1 + \dfrac{c_i(s)}{M_i s + D_i}\right)^{-1} \; .
\end{equation} 

\begin{figure}
\centering
\includegraphics[width=.7\columnwidth]{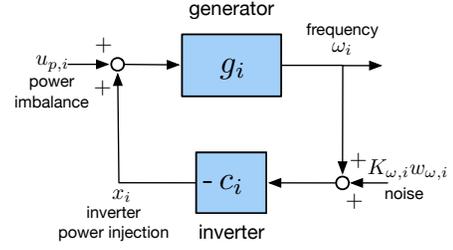}
\caption{Bus Dynamics $p_i$}\label{fig:2}
\end{figure}

\subsubsection*{Network Dynamics}
The network dynamics are given by
\begin{equation}\label{eq:N}
\hat p_e(s) = N(s) \hat u_N(s),\;\text{ where }\; N(s) = \frac{1}{s}L_B,
\end{equation}
$L_B \in \mathbb{R}^{n \times n}$ is the $B_{ij}$-weighted Laplacian matrix that describes how the transmission network couples the dynamics of different buses, i.e. $(L_B)_{ij}$ is equal to $-B_{ij}$, if $ij\in E$, $\sum_{e\in E} B_{e}$, if $i\!=\!j$, and $0$ otherwise, 
where $E$ is the set of buses.

% % \begin{subequations} 
%     \begin{equation}\label{eq:07}
%        \dot{\theta} =  u_N, \quad y_N = L_B \theta 
%     \end{equation}
% % \end{subequations} 
% where $\theta = (\theta_i, i \in V)$ denotes the vector of system phase, i.e. $\theta_i$ is the phase of bus i $u_N := \omega = (\omega_i, i \in V)$ the vector of system frequency deviation from nominal, $y_N = (y_{N,i}, i \in V)$ the vector of electric power network demand, i.e. $y_{N,i}$ is the total electric power drained by the network at bus i, and 

\subsubsection*{State-Space Model}
It will be useful to use the state space representation of the closed loop system of Figure \ref{fig:1}. Using time domain versions of \eqref{eq:gi}, \eqref{eq:swing}, \eqref{eq:ci}, and \eqref{eq:N}, the power system dynamics can be represented by 
\begin{subequations}\label{eq:state}
\begin{align}
\dot{\theta} 	&= \omega\\
M\dot \omega 	&= -D\omega - L_B\theta + x +p^{in}\\
\dot{x} 		&= -K_\delta(R_r^{-1}\omega + x) -K_\nu\dot \omega \label{eq:idroop}
\end{align}
\end{subequations}
where $M := \diag{M_i, i \in V }$, $R_r := \diag{R_{r,i}, i \in V}$,  $D := \diag{D_i, i \in V }$, $K_\delta:= \diag{\delta_i,i\in V}$, $K_\nu:= \diag{\nu_i,i\in V}$, $\theta := (\theta_i, i \in V)$, $\omega := (\omega_i, i \in V)$, $x := (x_i, i \in V)$, and $p^{in}:=(p^{in}_i,i\in V)$; we refer the reader to  \cite{pm2017ifac-wc} for more details connecting the two models.

Equation \eqref{eq:state} illustrates one of the potential challenges of implementing iDroop (or similarly virtual inertia), which comes from the need to measure $\dot \omega$ to implement \eqref{eq:idroop}. Fortunately, in the case of iDroop this can be avoided by considering the change of variable 
\begin{equation} 
    x = z - K_{\nu}\omega. \label{eq:2}
\end{equation}
Thus combining \eqref{eq:state} with \eqref{eq:2} gives
\begin{subequations} \label{eq:state-2}
    \begin{align}
        \dot{\theta} &= \omega\,, \label{eq:4-a}\\
        M \dot{\omega} &= -  D\omega-L_B\theta + (z - K_{\nu}\omega) +p^{in}\,,\label{eq:4-b}\\
        \dot{z} &= K_{\delta}(K_{\nu} - R_r^{-1})\omega - K_{\delta} z\,. \label{eq:4-c}
    \end{align}
\end{subequations}
To simplify our analysis we translate the equilibrium point of \eqref{eq:state-2} to the origin. Since \eqref{eq:state-2} is a linear time invariant system, this is equivalent to setting $p^{in}$ to zero. Thus in the rest of this paper, we assume without loss of generality that $p^{in}=0$.

\subsection{Performance Measures}\label{ssec:perfomrance-measures}
As mentioned before the aim of this paper is to evaluate how power disturbances, frequency measurement noise and delay affect overall dynamic behavior,  and in particular the frequency fluctuations, of power network described in Figure \ref{fig:1}.
To this end we define specific performance metrics that will enable us to show that grid-connected inverters implementing iDroop can improve the performance beyond existing solutions.

\subsubsection*{Measurements Noise and Power Disturbances}

We assume that the power system is being perturbed by two different signals (see figures \ref{fig:1} and \ref{fig:2}): (i) $K_pw_p\in\R^n$, with $K_p=\diag{k_{p,i},i\in V}$, that represents the fluctuations of the power injections at each bus; and (ii) $K_{\omega} w_\omega\in\R^n$, where $K_{\omega}=\diag{k_{\omega,i},i\in V}$, that represents the frequency measurements noise that inverters experience at each bus. The signals $w_p$ and $w_\omega$ represent uncorrelated stochastic white noise with zero mean and unit variance, i.e., $E[w_p(t)^Tw_p(\tau)]=\delta(t-\tau)I_n$ and $E[w_\omega(t)^Tw_\omega(\tau)]=\delta(t-\tau)I_n$.

%%%
Substituting $\omega$ with $\omega + K_\omega w_\omega$ on the RHS of \eqref{eq:4-b} and \eqref{eq:4-c}, $p^{in}=0$ with $K_pw_p$ in \eqref{eq:4-b}, and defining $y=\omega$ as the output of \eqref{eq:state-2} gives
%\begin{subequations} 
    \begin{equation}\label{eq:mimo}
       \begin{bmatrix} \dot{\theta} \\ \dot{\omega} \\ \dot{z} \end{bmatrix}
       =
       A \begin{bmatrix} \theta \\ \omega \\ z \end{bmatrix} + B \begin{bmatrix} w_p \\ w_{\omega} \end{bmatrix} \; ,\qquad
        y = C \begin{bmatrix} \theta \\ \omega \\ z \end{bmatrix}  ,
     \end{equation}      
%\end{subequations}
where
\begin{subequations} \label{eq:9}
    \begin{align}
       A &\!=\! \begin{bmatrix} 0_{n\times n} & I_n & 0_{n\times n} \\ -M^{-1}L_B & -M^{-1}(D + K_{\nu}) & M^{-1} \\ 0_{n\times n} & K_{\delta} (K_{\nu} - R_r^{-1}) & - K_{\delta} \end{bmatrix},\\
       B &\!=\! \begin{bmatrix} 0_{n\times n} \!&\! 0_{n\times n} \\ M^{-1}K_p \!&\! - M^{-1} K_{\nu} K_{\omega} \\ 0_{n\times n} \!&\! K_\delta(K_{\nu} - R_r^{-1})K_{\omega} \end{bmatrix}\!,\;
       C \!=\! \begin{bmatrix} 0_{n\times n} \!&\! I_n \!&\! 0_{n\times n} \end{bmatrix}\!.
    \end{align}
\end{subequations}

Thus if we let $G_\text{iDroop}$ denote the LTI system \eqref{eq:mimo}, we measure the effect of power disturbances and frequency measurements noise using the $H_2$ norm of the $G_\text{iDroop}$ which is given by 
\begin{flalign}\label{eq:dyn-cost}
\|G_\text{iDroop}\|_{H_2}^2 &= \lim_{t\rightarrow\infty} E [y^T(t)y(t)].
 \end{flalign}

The computation of the $H_2$ norm has been widely studied in modern control theory. In particular, one very useful procedure to compute $\|G_\text{iDroop}\|_{H_2}$ (see ~\cite{Doyle:1989kf}) is based on using 
\begin{equation}\label{eq:H2}
\|G_\text{iDroop}\|_{H_2}^2=\tr (B^TXB)
\end{equation}
where $X$ is the observability Grammian, i.e. $X$ solves the Lyapunov equation 
\begin{equation}\label{eq:lyapunov}
A^TX+XA=-C^TC.
\end{equation}

In the context of power systems this methodology has been first used in \cite{Tegling:2015ef}, where the authors seek to compute the power losses incurred by the network in the process of resynchronizing   generators after a disturbance.  Since then, several works have used similar metrics to evaluate effect of controllers on the power system performance, see e.g. \cite{Poolla:2015vq,Tegling:2016wna}.%\addcites.

\subsubsection*{Delay Robustness}

The frequency measurements available to the inverters will be subject to delay. When conducting inverter design, especially when considering inverters with fast dynamics and large gains, it is important to directly consider the effect of this delay to maintain the desired performance and prevent instability. 

Delays to the frequency measurements can be easily accommodated into the models discussed in Section~\ref{ssec:model} by multiplying the inverter transfer fuction models $c_i\s$ by $e^{-s\tau_i}$. For example, the iDroop contoller in \eqref{eq:05} now becomes
\[
        c_i(s) = \dfrac{\nu_i s + \delta_i R_{r,i}^{-1}}{s + \delta_i}e^{-s\tau_i}. 
\]
Here the constant $\tau_i>0$ corresponds to the delay to \emph{i}th inverters frequency measurment. We then define the delay robustness $\tau_\text{rob}$ of power network model to be the largest $\tau>0$ such that if
\[
\tau_i<\tau, \;\forall{}i\in{}V,
\]
then the model remains stable. The delay robustness therefore quanitifies how much delay to the frequency measurements can be tolerated before stability of the network is lost. This measure is in essence an adaptation of the phase margin to the our power system network model, and provides a classical measure of robustness to assess the designed inverters.

%%%%%%%%%%%%%%%%%%%%%%%%%%%%%%%%%%%%%%%%%%%%%%%%%%%%%%%%%%%%%%%%%%%%%%%%%%%%%%%%

%%%%%%%%%%%%%%%%%%%%%%%%%%%%%%%%%%%%%%%%%%%%%%%%%%%%%%%%%%%%%%%%%%%%%%%%%%%%%%%%
% Performance
%%%%%%%%%%%%%%%%%%%%%%%%%%%%%%%%%%%%%%%%%%%%%%%%%%%%%%%%%%%%%%%%%%%%%%%%%%%%%%%%
%!TEX root = main.tex
\section{Effect of Measurements Noise and Power Disturbances}

In this section we evaluate how the relative intensity of the two type of stochastic disturbances (frequency measurement noise and power disturbances) affect the overall system performance using the $H_2$ norm metric $||G_\text{iDroop}||_{H_2}$ described in Section \ref{ssec:perfomrance-measures}.
To make the analysis tractable, we assume homogeneous parameters $M_i=m$, $D_i=d$, $R_{r,i}=r_r$, $K_{p,i}=k_p$, $K_{\omega,i}=k_\omega$, $\delta_i=\delta$, and $\nu_i=\nu$, $\forall i\in V$.

In \cite{m2016cdc} it was shown that, if instead of iDroop, the inverters implement either Droop Control (DC) or virtual inertial (VI) then the $H_2$ norm is respectively given by 
\begin{equation}\label{eq:h2-dc}
\|G_\text{DC}\|_{H_2}^2 = \frac{n[k_p^2+(r_r^{-1}k_\omega)^2]}{2m(d+r_r^{-1})}\;\, \text{and}\;\,\|G_\text{VI}\|_{H_2}^2=+\infty,
\end{equation}
where $G_\text{DC}$ and $G_\text{VI}$ represent the system \eqref{eq:mimo} when the inverters implement droop control or virtual inertia respectively. In this section we compare \eqref{eq:h2-dc} with the corresponding formulae for iDroop that is computed in Section \ref{ssec:h2-idroop} and show in Section \ref{ssec:performace-improvement} that iDroop can use its additional flexility to  outperform DC and VI. 

\subsection{$H_2$-norm Performance of iDroop}\label{ssec:h2-idroop}

We now compute the $H_2$ norm of iDroop. Since we assume homogeneous parameters we can decouple \eqref{eq:mimo} into $n$ subsystems. More precisely, let $U$ be the orthonormal transformation that diagonalizes $L_B$. That is, $U^TU=UU^T=I_n$ such that $L_B=U\Gamma U^T$ where $\Gamma = \diag{\lambda_1, \ldots, \lambda_n}$ with $\lambda_i$ being the $i$th eigenvalue of $L_B$ ordered in increasing order ($\lambda_1=0\leq \lambda_2\leq\dots\leq\lambda_n$).
Then, making the following change of variables $\theta = U \theta'$, $\omega = U \omega'$, $z = U z'$, $y = U y'$, $w_p = U w_p'$, $w_{\omega} = U w_{\omega}'$
%, where $U$ is the orthonormal transformation that diagonalizes $L_B$, i.e. $U^{T} L_B U = \Gamma$ such that $\Gamma = \mathrm{diag} \{ \lambda_1, \ldots, \lambda_n \}$. Here, $U$ must exist since $L_B$ is symmetric and any symmetric matrix is diagonalizable via orthogonal matrix. Then, using the property of orthogonal matrix that $U^{-1} = U^{T}$ 
leads to
\begin{subequations} \label{eq:14}
    \begin{align}
       \begin{bmatrix} \dot{\theta}' \\ \dot{\omega}' \\ \dot{z}' \end{bmatrix} 
       &=\begin{bmatrix} 
       0_{n\times n} & I_n & 0_{n\times n} \\ -\frac{1}{m}\Gamma & -\frac{d+\nu}{m}I_n & \frac{1}{m}I_n \\ 0_{n\times n} & {\delta} ({\nu} - r_r^{-1})I_n & - {\delta}I_n  \end{bmatrix} \begin{bmatrix} \theta' \\ \omega' \\ z' \end{bmatrix} \nonumber\\ 
       &+ \begin{bmatrix} 0_{n\times n} & 0_{n\times n} \\ \frac{k_p}{m}I_n & -\frac{\nu k_{\omega}}{m}I_n \\ 0_{n\times n} & {\delta}({\nu} - r_r^{-1})k_{\omega}I_n \end{bmatrix} \begin{bmatrix} w_p
       ' \\ w_{\omega}' \end{bmatrix} \; ,
    \end{align}
    \begin{equation}
       y' = \begin{bmatrix} 0_{n\times n} & I_n & 0_{n\times n} \end{bmatrix} \begin{bmatrix} \theta' \\ \omega' \\ z' \end{bmatrix} \; .
    \end{equation}
\end{subequations}

Since \eqref{eq:14} is composed by diagonal matrices,
 it is equivalent to $n$ decoupled subsystems of the form
\begin{align}\label{eq:15}
        \begin{bmatrix} \dot{\theta}_i' \\ \dot{\omega}_i' \\ \dot{z}_i' \end{bmatrix}
       = A_i \begin{bmatrix} \theta_i' \\ \omega_i' \\ z_i' \end{bmatrix} + B_i \begin{bmatrix} w_{p,i}' \\ w_{\omega,i}' \end{bmatrix} \; ,   \quad
         y_i' = C_i \begin{bmatrix} \theta_i' \\ \omega_i' \\ z_i' \end{bmatrix} \; ,
\end{align}
where
\begin{subequations} \label{eq:16}
    \begin{align}
       A_i &= \begin{bmatrix} 0 & 1 & 0 \\ -\lambda_i/m & -(d + \nu )/m & 1/m \\ 0 & \delta(\nu - r_r^{-1}) & -\delta \end{bmatrix} \; ,\label{eq:16a}\\
       B_i &= \begin{bmatrix} 0 & 0 \\ k_p/m & -\nu k_{\omega}/m \\ 0 & \delta (\nu - r_r^{-1})k_{\omega} \end{bmatrix} , \;
       C_i = \begin{bmatrix} 0 & 1 & 0 \end{bmatrix} \; .\label{eq:16b}
    \end{align}
\end{subequations}

Therefore, since $U$ is an orthonormal transformation, it does not affect the $H_2$ norm of \eqref{eq:mimo}, and it will allow us to compute the $H_2$ norm of \eqref{eq:mimo} by making $n$ simpler $H_2$ norm computations using  \eqref{eq:15}.

\begin{theorem}[$H_2$ norm of iDroop]\label{th:h2-idroop}
The $H_2$ norm of the system \eqref{eq:mimo}, i.e. $||G_\text{iDroop}||_{H_2}$ is given by 
\begin{align}
    & \|G_\text{iDroop}\|^2_{H_2}= \dfrac{n(k_p^2+\nu^2 k_\omega^2)}{2m(d + \nu)}\nonumber\\
    & + \sum_{i=1}^n{\dfrac{\delta^2(\nu - r_r^{-1})\left[\dfrac{k_p^2+\nu^2 k_{\omega}^2}{d+\nu} - (\nu+ r_r^{-1})k_\omega^2\right] }{2\left[(d+\nu + m\delta)\delta(d + r_r^{-1}) + (d+\nu)\lambda_i\right]}}\; . \label{eq:h2-idroop}
\end{align}\vspace{1ex}
\end{theorem}
In some cases, it will be convenient to make explicit the dependence of the $H_2$ norm of iDroop on the parameters $\nu$ and $\delta$. Thus we will also use $\|G_\text{iDroop}\|^2_{H_2}(\nu,\delta)$ to refer to \eqref{eq:h2-idroop}.

\begin{proof}
As mentioned before, we compute $||G_\text{iDroop}||_{H_2}$ by computing instead the norm of $n$ orthogonal subsystems $G_{\text{iDroop},i}$ described by \eqref{eq:15}. 
Therefore we need to solve the Lyapunov equation
\begin{equation}
   A^T_i Q + Q A_i = - C^T_i C_i \; , \label{eq:17}
\end{equation}
where $A_i$ is given by \eqref{eq:16a} and 
 $Q$ must be symmetric, i.e.,
\begin{equation}
   Q = \begin{bmatrix} q_{11} & q_{12} & q_{13} \\ q_{12} & q_{22} & q_{23} \\ q_{13} & q_{23} & q_{33} \end{bmatrix} \; . \label{eq:18}
\end{equation}
Whenever $\lambda_i\not=0$ \eqref{eq:17} has a unique solution for $Q$. However, when $\lambda_i=0$ the system \eqref{eq:15} has a zero pole which could render infinite $H_2$ norm. However, as we will later see this mode is unobservable. 

We focus first on the case $\lambda_i\not=0$. Thus we have:
\begin{subequations} \label{eq:19}
\begin{align}
   & QA_i =  \begin{bmatrix} q_{11} & q_{12} & q_{13} \\ q_{12} & q_{22} & q_{23} \\ q_{13} & q_{23} & q_{33} \end{bmatrix} \begin{bmatrix} 0 & 1 & 0 \\ -\lambda_i/m & -(d+\nu)/m & 1/m \\ 0 & \delta(\nu - r_{r}^{-1}) & -\delta \end{bmatrix} \nonumber\\
 & \!=\!  \begin{bmatrix} \!-\!\dfrac{\lambda_i}{m}q_{12} & q_{11} \!-\!a_1q_{12} \!+\! \delta a_2q_{13} & \dfrac{1}{m}q_{12} \!-\!\delta q_{13} \\ \!-\!\dfrac{\lambda_i}{m}q_{22} & q_{12} \!-\!a_1q_{22} \!+\! \delta a_2q_{23} & \dfrac{1}{m}q_{22} \!-\!\delta q_{23} \\ \!-\!\dfrac{\lambda_i}{m}q_{23} & q_{13} \!-\!a_1q_{23} + \delta a_2q_{33} & \dfrac{1}{m}q_{23} \!-\!\delta q_{33} \end{bmatrix}\!\!,\label{eq:19a}
\end{align}
where $a_1=\frac{d+\nu}{m}$ and $a_2=\nu - r_r^{-1}$. Similarly, 
\begin{equation}\label{eq:19b}
A^T_i Q =A^T_i Q^T = (QA_i)^T
\end{equation}
which is the transpose of \eqref{eq:19a}.
% \begin{align}
%     A^T_i Q = & A^T_i Q^T = (QA_i)^T \nonumber\\
%     = & \begin{bmatrix} -\dfrac{\lambda_i}{m}q_{12} & q_{11} -\dfrac{\nu + d}{m}q_{12} + \delta(\nu - r_r^{-1})q_{13} & \dfrac{1}{m}q_{12} -\delta q_{13} \\ -\dfrac{\lambda_i}{m}q_{22} & q_{12} -\dfrac{\nu + d}{m}q_{22} + \delta(\nu - r_r^{-1})q_{23} & \dfrac{1}{m}q_{22} -\delta q_{23} \\ -\dfrac{\lambda_i}{m}q_{23} & q_{13} -\dfrac{\nu + d}{m}q_{23} + \delta(\nu - r_r^{-1})q_{33} & \dfrac{1}{m}q_{23} -\delta q_{33} \end{bmatrix}^T \; ,
% \end{align}
Finally,
\begin{equation}
   - C^T_i C_i = - \begin{bmatrix} 0 \\ 1 \\ 0 \end{bmatrix} \ \begin{bmatrix} 0 & 1 & 0 \end{bmatrix} = \begin{bmatrix} 0 & 0 & 0 \\ 0 & -1 & 0 \\ 0 & 0 & 0 \end{bmatrix} \; .
\end{equation}
\end{subequations}

Substituting (\ref{eq:19}) into (\ref{eq:17})  gives 
\begin{subequations} \label{eq:20}
    \begin{equation}
    -\dfrac{2\lambda_i}{m}q_{12} = 0 \label{eq:20a} \; ,
    \end{equation}
    \begin{equation}
    q_{11} -a_1q_{12} + \delta a_2q_{13} -\dfrac{\lambda_i}{m}q_{22} = 0 \; ,
    \end{equation}
    \begin{equation}
    \dfrac{1}{m}q_{12} -\delta q_{13} -\dfrac{\lambda_i}{m}q_{23} = 0 \; ,
    \end{equation}
    \begin{equation}
    2\left(q_{12} -a_1q_{22} + \delta a_2q_{23}\right) = -1 \; ,
    \end{equation}
    \begin{equation}
    \dfrac{1}{m}q_{22} -\delta q_{23} + q_{13} -a_1q_{23} + \delta a_2q_{33} = 0 \; ,
    \end{equation}
    \begin{equation}
    2(\dfrac{1}{m}q_{23} -\delta q_{33}) = 0 \; .
    \end{equation}
\end{subequations}

Since $\lambda_i\ne0$,  from (\ref{eq:20a}) we get
\begin{equation}
    q_{12} = 0 \; . \label{eq:21}
\end{equation}

Substituting (\ref{eq:21}) into (\ref{eq:20}) gives simplified equations
\begin{subequations} \label{eq:22}
    \begin{equation}
    q_{11} + \delta a_2q_{13} -\dfrac{\lambda_i}{m}q_{22} = 0 \; , \label{eq:22a}
    \end{equation}
    \begin{equation}
    q_{13} = -\dfrac{\lambda_i}{m\delta}q_{23} \; , \label{eq:22b}
    \end{equation}
    \begin{equation}
    2\left(-a_1q_{22} + \delta a_2q_{23}\right) = -1 \; , \label{eq:22c}
    \end{equation}
    \begin{equation}
    \dfrac{1}{m}q_{22} + q_{13} -(a_1 + \delta)q_{23} + \delta a_2q_{33} = 0 \; , \label{eq:22d}
    \end{equation}
    \begin{equation}
    q_{23} = m\delta q_{33} \; . \label{eq:22e}
    \end{equation}
\end{subequations}

Substituting (\ref{eq:22e}) into (\ref{eq:22b}) and (\ref{eq:22c}), respectively, and using again $a_1=\frac{d+\nu}{m}$ and $a_2=\nu - r_r^{-1}$ leads to
\begin{equation}
    q_{13} = -\lambda_i q_{33} \; , \label{eq:23}
\end{equation}
\begin{equation}
    q_{22} = \dfrac{m}{d+\nu} \left[m\delta^2(\nu - r_r^{-1}) q_{33} + \dfrac{1}{2}\right] \; . \label{eq:24}
\end{equation}

Substituting (\ref{eq:23}) and (\ref{eq:24}) into  (\ref{eq:22a}) gives
\begin{equation}
    q_{11} = \dfrac{\lambda_i}{d+\nu} \left[m\delta^2(\nu - r_r^{-1}) q_{33} + \dfrac{1}{2}\right] + \delta(\nu - r_r^{-1})\lambda_i q_{33} \; .
\end{equation}

Substituting (\ref{eq:22e}), (\ref{eq:23}), and (\ref{eq:24}) into  (\ref{eq:22d}), respectively, we can get
\begin{align}
  0&= \dfrac{1}{d+\nu} \left[m\delta^2(\nu - r_r^{-1}) q_{33} +\dfrac{1}{2}\right] -\lambda_i q_{33} \\
  &-\left(\dfrac{d+\nu}{m} + \delta\right)m\delta q_{33} + \delta(\nu - r_r^{-1})q_{33} \; ,
\end{align}
from which $q_{33}$ can be solved as
\begin{align}
   q_{33} = \dfrac{1}{2\left[(d+\nu + m\delta)\delta(d + r_r^{-1}) + (d+\nu)\lambda_i\right]} \; . \label{eq:27}
\end{align}

Now, we can compute $\|G_{\text{iDroop},i}\|^2_{H_2}$ based on 
\begin{equation}
    \|G_{\text{iDroop},i}\|^2_{H_2} = \mathrm{tr}(B_i^T Q B_i) \; . \label{eq:28}
\end{equation}

Substituting (\ref{eq:16b}) and (\ref{eq:18}) into  (\ref{eq:28}), we can get
\begin{align}
    & \|G_{\text{iDroop,i}}\|^2_{H_2}= \dfrac{k_p^2+\nu^2 k_{\omega}^2}{m^2} q_{22} \nonumber\\
    &  - \dfrac{2\nu k_{\omega}^2 }{m}\delta (\nu - r_r^{-1}) q_{23} + \delta^2 (\nu - r_r^{-1})^2k_{\omega}^2 q_{33} \; . \label{eq:29}
\end{align}

Substituting (\ref{eq:22e}) and (\ref{eq:24}) into (\ref{eq:29}) gives
\begin{align}
    & \|G_{\text{iDroop},i}\|^2_{H_2} = \dfrac{k_p^2+\nu^2 k_{\omega}^2}{2m(d+\nu)} \nonumber\\
    & + \delta^2(\nu - r_r^{-1})\left[\dfrac{k_p^2+\nu^2 k_{\omega}^2}{d+\nu} - (\nu + r_r^{-1})k_{\omega}^2\right] q_{33}  \; . \label{eq:30}
\end{align}

Substituting (\ref{eq:27}) into (\ref{eq:30}) gives
\begin{align}
    & \|G_{\text{iDroop},i}\|^2_{H_2}= \dfrac{k_p^2+\nu^2 k_\omega^2}{2m(d+\nu)}\nonumber\\
    & + \dfrac{\delta^2(\nu - r_r^{-1})\left[\dfrac{k_p^2+\nu^2 k_{\omega}^2}{d+\nu} - (\nu+ r_r^{-1})k_\omega^2\right] }{2\left[(d+\nu + m\delta)\delta(d + r_r^{-1}) + (d+\nu)\lambda_i\right]}\; . \label{eq:31}
\end{align}

We now consider the case $\lambda_i=0$, i.e. $i=1$. Since $\lambda_1=0$, neither $\dot\omega'_1$, nor $\dot z'_1$, nor $y'_1$ depend on $\theta'_1$ in \eqref{eq:15}. Thus $\theta_i'$ is not observable and we can simplify the subsystem described by (\ref{eq:15}) to
% \begin{subequations} 
    \begin{align}\label{eq:32}
        \begin{bmatrix} \dot{\omega}_1' \\ \dot{z}_1' \end{bmatrix}
       =
       & \tilde A_1\begin{bmatrix} \omega_1' \\ z_1' \end{bmatrix} 
        + \tilde B_1\begin{bmatrix} w_{p,1}' \\ w_{\omega,1}' \end{bmatrix} \; ,\quad   
%     \end{align}
%     \begin{equation}
        y_1' = \tilde C_1 \begin{bmatrix} \omega_1' \\ z_1' \end{bmatrix} \;,
    \end{align}
%\end{subequations}
where, for the simplified subsystem, we have
\begin{subequations}
    \begin{align}
        \tilde A_1 &= \begin{bmatrix} -(d+\nu)/m & 1/m \\ \delta(\nu - r_r^{-1}) & -\delta \end{bmatrix} \;,\\
%         \end{equation}
%      \begin{equation}
           \tilde B_1 &= \begin{bmatrix} k_p/m & -\nu k_{\omega}/m \\ 0 & \delta (\nu - r_r^{-1})k_{\omega} \end{bmatrix} \;, 
  \quad
        \tilde C_1 = \begin{bmatrix} 1 & 0 \end{bmatrix}  \;.\label{eq:33b}
     \end{align}
\end{subequations}
% \begin{subequations}
%     \begin{equation}
%         \tilde A_1 = \begin{bmatrix} -(\nu + d)/m & 1/m \\ \delta(\nu - r_r^{-1}) & -\delta \end{bmatrix} \;,
%      \end{equation}
%      \begin{equation}
%         \tilde B_1 = \begin{bmatrix} k_p/m & -\nu k_\omega/m \\ 0 & \delta (\nu - r_r^{-1})k_\omega \end{bmatrix} \;, \label{eq:33b}
%      \end{equation}
%      \begin{equation}
%         \tilde C_1 = \begin{bmatrix} 1 & 0 \end{bmatrix}  \;.
%      \end{equation}
% \end{subequations}

Again, we need to solve the Lyapunov equation shown in (\ref{eq:17}), but here we use 
\begin{equation}
   Q = \begin{bmatrix} q_{11} & q_{12}\\ q_{12} & q_{22} \end{bmatrix} \; . \label{eq:34}
\end{equation}

Thus a similar computation as for the case $\lambda_i\not=0$ gives
\begin{align}
    & \|G_{\text{iDroop},i}\|^2_{H_2} = \dfrac{k_p^2+\nu^2 k_{\omega}^2}{2m(d+\nu)}\;  \nonumber\\
    & + \dfrac{\delta^2(\nu - r_r^{-1})\left[\dfrac{k_p^2+\nu^2 k_{\omega}^2}{d+\nu} - (\nu + r_r^{-1})k_{\omega}^2\right] }{2(d+\nu + m\delta)\delta(d + r_r^{-1})}. \label{eq:42}
\end{align}

Comparing (\ref{eq:31}) and (\ref{eq:42}), it is easy to find that if $\lambda_i$ in (\ref{eq:31})  equals to zero, (\ref{eq:31}) becomes (\ref{eq:42}). So we can conclude that no matter $\lambda_i$ equals to zero or not, $\|G_{\text{iDroop},i}\|^2_{H_2}$ can be represented by (\ref{eq:31}).
Therefore, we have $\|G_\text{iDroop}\|^2_{H_2} = \sum_{i=1}^n\|G_{\text{iDroop},i}\|^2_{H_2}$ which gives \eqref{eq:h2-idroop}.
% \begin{equation}
%     \|H_{iDroop}\|^2_{H_2} = \sum_{i=1}^n\|H_{iDroop,i}\|^2_{H_2} \; . 
% \end{equation}
\end{proof}

%\subsection{Low $\delta$ Region Simplifies the Situation}

%Since we can always choose a sufficiently small $\delta$ such that $f(\delta) \simeq n \alpha_5$, if we can find a $\nu$ such that $\alpha_5 < \|H_{DC,i}\|^2_{H_2}$, then iDroop will perform better than DC with that $\nu$ and a sufficiently small $\delta$.

\subsection{iDroop Performance Improvement}\label{ssec:performace-improvement}
We now show using Theorem \ref{th:h2-idroop} that iDroop can adapt to different network conditions and improve the system performance. Since virtual inertia has infinite $H_2$ norm (c.f. \eqref{eq:h2-dc}), we compare here the performance of iDroop, i.e., $\|G_{\text{iDroop}}\|_{H_2}$, with inverters implementing standard droop control, i.e., $\|G_\text{DC}\|_{H_2}$.
We achieve this goal by finding the set of parameter values $\nu$ and $\delta$ that minimizes \eqref{eq:h2-idroop} and showing that for such values $\|G_\text{iDroop}\|_{H_2}\leq \|G_\text{DC}\|_{H_2}$. The following auxiliary lemmas help us pave the way to the main result of this section (Theorem \ref{th:h2-improves}).

\begin{lemma} \label{lem:1}
If $\delta \to +\infty$, then $\|G_\text{iDroop}\|^2_{H_2} = \|G_\text{DC}\|^2_{H_2}$.
\end{lemma}

\begin{proof}
When $\delta \to +\infty$, (\ref{eq:31}) can be reduced as
\begin{align}
   &\underset{\delta\rightarrow+\infty}{\lim} \|G_{\text{iDroop},i}\|^2_{H_2} \nonumber\\
    & =\dfrac{k_p^2 + \nu^2 k_\omega^2}{2m(d+\nu)} + \dfrac{(\nu - r_r^{-1}) \{ k_p^2 -\left[\nu d + r_r^{-1}(d+\nu)\right] k_\omega^2 \}  }{2m (d+\nu) (d + r_r^{-1})}  \nonumber\\
%& = \dfrac{(d+\nu) k_p^2 + \{\nu^2 (d + r_r^{-1}) - (\nu - r_r^{-1}) \left[\nu d + r_r^{-1}(d+\nu)\right] \} k_\omega^2}{2m(d+\nu)(d + r_r^{-1})} \nonumber\\
%     & = \dfrac{k_p^2}{2m(d + r_r^{-1})} \nonumber\\
%     & + \dfrac{ \{\nu^2 (d + r_r^{-1}) - (\nu - r_r^{-1}) \left[\nu ( d + r_r^{-1}) + r_r^{-1}d\right] \} k_\omega^2}{2m(\nu + d)(d + r_r^{-1})} \nonumber\\
    & = \dfrac{k_p^2}{2m(d + r_r^{-1})} + \dfrac{ \left[\nu r_r^{-1} (d + r_r^{-1}) - (\nu - r_r^{-1}) r_r^{-1}d\right] k_\omega^2}{2m(d+\nu)(d + r_r^{-1})} \nonumber\\
    & = \dfrac{k_p^2}{2m(d + r_r^{-1})} + \dfrac{ \left(\nu r_r^{-2} + r_r^{-2}d\right) k_\omega^2}{2m(d+\nu)(d + r_r^{-1})} \nonumber\\
    & = \dfrac{k_p^2+\left(r_r^{-1}k_\omega \right)^2}{2m(d + r_r^{-1})} =: \|G_{\text{DC},i}\|^2_{H_2}\;, \label{eq:44}
\end{align}
where $\|G_{\text{DC},i}\|_{H_2}$ denotes the $H_2$ norm of the $i$th subsystem \eqref{eq:15} when droop control with gain $r_r$ is used instead of iDroop. Thus it follows that $\|G_\text{DC}\|_{H_2}^2=\sum_{i=1}^n\|G_{\text{DC},i}\|_{H_2}^2$ which gives \eqref{eq:h2-dc}.
\end{proof}

Thus we have shown that as $\delta \to +\infty$ $\|G_\text{iDroop}\|_{H_2}$ asymptotically converges to $\|G_\text{DC}\|_{H_2}$. Our second lemma will show that this convergence is monotonically from either above or below depending on the value of $\nu$.
For a given $\nu$, we use $f(\delta)$ to denote the dependence of $\|G_\text{iDroop}\|^2_{H_2}$ with respect to $\delta$, i.e. 
\begin{equation}\label{eq:f}
f(\delta) =\sum_{i=1}^n f_i(\delta) = n\alpha_5 +\sum_{i=1}^n  \dfrac{\alpha_1 \delta^2}{\alpha_2 \delta^2+\alpha_3 \delta +\alpha_4(\lambda_i)} \
    \end{equation}
with
\begin{subequations} \label{eq:49}
    \begin{align}
        \alpha_1 &= (\nu - r_r^{-1})\left[\dfrac{k_p^2+\nu^2 k_\omega^2}{d+\nu} - (\nu + r_r^{-1})k_\omega^2\right]\;,\label{{eq:alpha1}}\\
        \alpha_2 &= 2m(d + r_r^{-1})\;,\;
        \alpha_3 = 2(d+\nu)(d + r_r^{-1})\;,   \\
        \alpha_4(\lambda_i) &= 2(d+\nu)\lambda_i\;,\quad\;
        \alpha_5 = \dfrac{k_p^2+\nu^2 k_\omega^2}{2m(d+\nu)}\;. 
    \end{align}    
\end{subequations}

\begin{lemma} \label{lem:2}
For any positive $\delta$, $f(\delta)$ is a monotonically increasing or decreasing function if and only if $\alpha_1$ is positive or negative, respectively. That is, $\mathrm{sign} \left( f'(\delta) \right) = \mathrm{sign} \left( \alpha_1 \right)$.
\end{lemma}

\begin{proof}
Using \eqref{eq:f} we compute the derivative of $f(\delta)$ which gives
\begin{align}
    f'(\delta)  % \nonumber\\
%    &= \sum_{i=1}^n\ \dfrac{2\alpha_1 \delta(\alpha_2 \delta^2+\alpha_3 \delta +\alpha_4(\lambda_i))-\alpha_1 \delta^2(2\alpha_2 \delta+\alpha_3)}{(\alpha_2 \delta^2+\alpha_3 \delta +\alpha_4(\lambda_i))^2} \ \nonumber\\
    &= \sum_{i=1}^n\ \alpha_1 \dfrac{ \alpha_3 \delta^2 +2 \alpha_4(\lambda_i) \delta}{(\alpha_2 \delta^2+\alpha_3 \delta +\alpha_4(\lambda_i))^2} \ \nonumber
\end{align}

From (\ref{eq:49}), we can easily see that $\alpha_2$ and $\alpha_3$ are positive since $m>0$, $d>0$, $\nu>0$, and $r_r>0$. Also, given that all the eigenvalues of the Laplacian matrix $L_B$ are nonnegative, $\alpha_4(\lambda_i)$ must be nonnegative. Thus, $\forall \delta >0$, $( \alpha_3 \delta^2 +2 \alpha_4(\lambda_i)\delta )/(\alpha_2 \delta^2+\alpha_3 \delta +\alpha_4(\lambda_i))^2 > 0$. So the sign of $f'(\delta)$ is determined by $\alpha_1$ for any $\delta > 0$.
\end{proof} 

Lemmas \ref{lem:1} and \ref{lem:2} suggest that in order to improve performance one needs to set $\nu$ such that $\alpha_1>0$ and $\delta$ as small as practically possible. The last lemma characterizes the optimal $\nu^*$ that minimizes the $H_2$ norm of iDroop when $\delta = 0$. 
%Under such scenario, $\|G_\text{iDroop}\|_{H_2}^2 = \|G_\text{DC}\|_{H_2}^2|_{r_r=\nu^{-1}}=n\alpha_5$.

\begin{lemma}\label{lem:3}
Let 
\begin{equation}\label{eq:g}
g(\nu) := \|G_\text{iDroop}\|_{H_2}^2(\nu,0)=\dfrac{n (k_p^2+\nu^2 k_\omega^2)}{2m(d+\nu)} 
\end{equation}
be the value of the $H_2$ norm of iDroop when $\delta=0$. Then, within the domain of interest $\nu\geq 0$, $g(\nu)$ is minimized by
\begin{equation}\label{eq:nu-star}
\nu^{\ast} = -d + \sqrt{d^2 + {k_p^2}/{k_\omega^2}}.
\end{equation}
\end{lemma}

\begin{proof}
We take the derivate of \eqref{eq:g} with respect to $\nu$ which gives
\begin{equation}
        g'(\nu) = \dfrac{k_\omega^2\nu^2 + 2 k_\omega^2 d\nu - k_p^2}{2m (d+\nu)^2}\;. \label{eq:43}
\end{equation}
By equating (\ref{eq:43}) to 0, we can solve the corresponding $\nu$ as $\nu^*_\pm=-d \pm \sqrt{d^2 + {k_p^2}/{k_\omega^2}}$. The only positive root is therefore $\nu^*=-d + \sqrt{d^2 + {k_p^2}/{k_\omega^2}}$.
%For simplicity, we denote this value as $\nu^{\ast}$. 
Since the denominator of \eqref{eq:43} is always positive and the highest order coefficient of the numerator is positive, whenever $0 < \nu < \nu^{\ast}$, then $ g'(\nu) < 0$, and if $\nu > \nu^{\ast}$, then $ g'(\nu) > 0$. Therefore, $\nu^{\ast}$ is actually the minimizer of $g(\nu)$ and $g({\nu^{\ast}})$ is the minimum.
\end{proof}

We are now ready to prove the main result of this section.
\begin{theorem}\label{th:h2-improves}
Let $\|G_\text{iDroop}\|^2_{H_2}(\nu,\delta)$ be the $(\nu,\delta)$-dependent $H_2$ norm of iDroop given in \eqref{eq:h2-idroop}. Let $\nu^*$ be as defined in \eqref{eq:nu-star}. 
Then, whenever $\left(k_p/k_\omega\right)^2 \neq 2r_r^{-1}d + r_r^{-2}$, setting 
\begin{equation}\label{eq:condition}
\nu \in [\nu^*,r_r^{-1}) \quad \text{or}\quad \nu \in (r_r^{-1},\nu^*],
\end{equation}
gives $\|G_\text{iDroop}\|^2_{H_2}(\nu,\delta)<\|G_\text{DC}\|^2_{H_2}$ for all $\delta \geq 0$. Moreover, $\|G_\text{iDroop}\|^2_{H_2}(\nu^*,0)$ provides the global minimum $H_2$ norm. When  $\left(k_p/k_\omega\right)^2 = 2r_r^{-1}d + r_r^{-2}$, then $\|G_\text{iDroop}\|^2_{H_2}(\nu^*=r_r^{-1},\delta)=\|G_\text{DC}\|^2_{H_2}$.
%Moreover, $\|H_{iDroop}\|^2_{H_2} \le \|H_{DC}\|^2_{H_2}$ whenever
\end{theorem}

\begin{proof}
By Lemma 2, for a given $\nu$, if $\alpha_1 < 0$, then $f'(\delta) < 0$, which follows that $\|G_\text{iDroop}\|^2_{H_2}$ always decreases as $\delta$ increases. However, according to Lemma 1, even if $\delta \to \infty$, we can only obtain $\|G_\text{iDroop}\|^2_{H_2} = \|G_\text{DC}\|^2_{H_2}$. Similarly, if $\alpha_1 = 0$, then $f'(\delta) = 0$, which indicates that $\|G_\text{iDroop}\|^2_{H_2}$ keeps constant as $\delta$ increases, so no matter what the value of $\delta$ we use, we will always obtain $\|G_\text{iDroop}\|^2_{H_2} = \|G_\text{DC}\|^2_{H_2}$. Therefore, iDroop control cannot outperform DC when $\alpha_1 \le 0$. Thus we constrain to  $\alpha_1 > 0$ from now on. In this case, $f'(\delta) > 0$, which shows that $\|G_\text{iDroop}\|^2_{H_2}$ always increases as $\delta$ increases, so choosing $\delta$ arbitrarily small is optimal for fix $\nu$. 

We now look at the values of $\nu$ that satisfy the requirement  $\alpha_1 > 0$. $\alpha_1$ can be rearranged as follows:
\begin{align}
	\alpha_1 %&= (\nu - r_r^{-1})\left[\dfrac{k_p^2+\nu^2 k_\omega^2}{d+\nu} - (\nu + r_r^{-1})k_\omega^2\right] \nonumber\\
%     &= \dfrac{(\nu - r_r^{-1}) \{ k_p^2-[r_r^{-1}(d+\nu)+\nu d] k_\omega^2 \} }{d+\nu} \nonumber\\
%     &= \dfrac{(\nu - r_r^{-1}) \left[( k_p^2 -  k_\omega^2 r_r^{-1}d) -k_\omega^2 ( d + r_r^{-1}) \nu  \right] }{d+\nu} \nonumber\\
     &= \dfrac{\beta_2 \nu^2 + \beta_1\nu +\beta_0}{d+\nu} \nonumber\;,
\end{align}
where $\beta_2=- k_\omega^2 ( d + r_r^{-1})$, $\beta_1=( k_p^2 + k_\omega^2   r_r^{-2} )$, and $\beta_0=- r_r^{-1} (k_p^2- k_\omega^2 r_r^{-1} d )$. Since the denominator of $\alpha_1$ is always positive, the sign of $\alpha_1$ is only decided by its numerator. 
Let $N_{\alpha_1}(\nu)$ be the numerator of $\alpha_1$. Thus, $N_{\alpha_1}(\nu)$ is a univariate quadratic function in $\nu$, whose roots are:
\begin{subequations}
	\begin{align}
	\nu_1 = r_r^{-1} \nonumber \;,\quad \text{and} \quad
% 	\end{align}
%     \begin{align}
    \nu_2 =  \dfrac{\left(k_p/k_\omega\right)^2 - r_r^{-1} d }{d + r_r^{-1}} \nonumber\;.
    \end{align}
\end{subequations}
Since $\beta_2 < 0$, the graph of $N_{\alpha_1}(\nu)$ is a parabola that opens downwards. Therefore, if $\nu_1 < \nu_2$, then $\nu \in (\nu_1, \nu_2)$ guarantees $\alpha_1 > 0$; if $\nu_1 > \nu_2$, then $\nu \in (\nu_2, \nu_1)\cap[0, \infty)$ guarantees $\alpha_1 > 0$. Notably, if $\nu_1 = \nu_2$, there exists no feasible points of $\nu$ to make $\alpha_1 > 0$. This can only happen if $\left(k_p/k_\omega\right)^2 = 2r_r^{-1} d + r_r^{-2}$, and in this case $\nu_1=\nu^*=\nu_2=r_r^{-1}$ and therefore
iDroop can only match the performance of DC by setting  $\nu=r_r^{-1}$ and choosing any $\delta\geq 0$, i.e., $\|G_\text{iDroop}\|^2_{H_2}(\nu^*=r_r^{-1},\delta)=\|G_\text{DC}\|^2_{H_2}$. This shows the last statement of the theorem.
%But given that it is not likely to happen in reality that $\nu_1$ is exactly equal to $\nu_2$, which requires that $\left(k_p/k_\omega\right)^2 = 2r_r^{-1} d + r_r^{-2}$, we deem that $\nu_1 \neq \nu_2$ in general.

We now focus in the case where the set $$S=(\nu_1, \nu_2)\cup\{(\nu_2, \nu_1)\cap[0, \infty)\}$$ is nonempty. Thus,  for any fix $\nu\in S$, $\alpha_1>0$, and thus by Lemma \ref{lem:2} setting $\delta=0$ achieves the minimum norm, i.e.  $\|G_\text{iDroop}\|^2_{H_2}(\nu,0)<\|G_\text{iDroop}\|^2_{H_2}(\nu,\delta)$ $\forall \delta>0$.
Since by Lemma 3, $\nu^{\ast}$ is the minimizer of $g(\nu)=\|G_\text{iDroop}\|^2_{H_2}(\nu,0)$, as long as $\nu^{\ast}\in S$, $(\nu^*,0)$ globally minimizes $\|G_\text{iDroop}\|^{2}_{H_2}(\nu,\delta)$. In fact, we will show next that $\nu^*$ is always within $S$ whenever $S\not=\emptyset$.

First we consider the situation where $\nu_1 < \nu_2$, which implies that $\left(k_p/k_\omega\right)^2 > 2r_r^{-1} d + r_r^{-2}$. Then we have
\begin{align}
	\nu^{\ast} &= -d+ \sqrt{ d^2 + {k_p^2}/{k_\omega^2}}\! >\!-d\!+\!\sqrt{d^2 + 2r_r^{-1}d + r_r^{-2}}\nonumber\\
    &=-d+\sqrt{( d+r_r^{-1})^2}=r_r^{-1} = \nu_1\nonumber\;.
\end{align}
We also want to show that $\nu^{\ast} < \nu_2$ which holds iff \\
$\sqrt{d^2 + {k_p^2}/{k_\omega^2}} < \dfrac{\left(k_p/k_\omega\right)^2 - r_r^{-1} d }{d + r_r^{-1}} + d = \dfrac{\left(k_p/k_\omega\right)^2 + d^2 }{d + r_r^{-1}} $\\
$\iff 1 < \dfrac{\sqrt{d^2 + {k_p^2}/{k_\omega^2}} }{d + r_r^{-1}}$, which always holds since  $\left(k_p/k_\omega\right)^2 > 2r_r^{-1} d + r_r^{-2}$.
%\\ 
%$\dfrac{\sqrt{d^2 + {k_p^2}/{k_\omega^2}} }{d + r_r^{-1}}> \dfrac{\sqrt{d^2 + 2r_r^{-1}d + r_r^{-2}} }{d + r_r^{-1}} %= \dfrac{\sqrt{(d + r_r^{-1})^2} }{d + r_r^{-1}} 
%	= 1.$\\
Thus, $\nu_1 < \nu^{\ast} < \nu_2$.
Similarly, we can prove that in the situation where $\nu_1 > \nu_2$, $\nu_2 < \nu^{\ast} < \nu_1$ holds and thus $\nu^{\ast} \in (\nu_2, \nu_1)\cap[0, \infty)$.
Thus, it follows that $(\nu^*,0)$ is the global minimizer of $\|G_\text{iDroop}\|^{2\ast}_{H_2}(\nu,\delta)$.

Finally, from \eqref{eq:h2-dc} and \eqref{eq:h2-idroop}, it is easy to see that 
$\|G_\text{DC}\|^{2}_{H_2}=\|G_\text{iDroop}\|^{2}_{H_2}(r_r^{-1},0)$.
From the proof of Lemma \ref{lem:3} it follows that when $\nu<\nu^*$ (resp. $\nu>\nu^*$)
then $\|G_\text{iDroop}\|^{2}_{H_2}(\nu,0)=g(\nu)$ decreases (resp. increases) monotonically.
Therefore, whenever $\nu$ satisfies \eqref{eq:condition} we have
$$\|G_\text{DC}\|^{2}_{H_2}=\|G_\text{iDroop}\|^{2}_{H_2}(r_r^{-1},0)>\|G_\text{iDroop}\|^{2}_{H_2}(\nu,0).$$
Result follows.
%If we can find some $\nu$ located between $\nu_1$ and $\nu_2$ such that $n\alpha_5 < \|G_\text{DC}\|^2_{H_2}$, then iDroop will perform better than DC with that $\nu$ and a sufficiently small $\delta \to 0$. We claim that $\nu^{\ast}$ is such a $\nu$, and in fact the best $\nu$ we can choose. It is easy to observe that actually $\|G_\text{DC}\|^2_{H_2} = n\alpha_5\lvert_{r_r^{-1}}$. Obviously, $n\alpha_5\lvert_{v^{\ast}} < n\alpha_5\lvert_{r_r^{-1}}$  whenever $r_r^{-1} \neq \nu^{\ast}$, since $v^{\ast}$ optimizes $n\alpha_5$. However, $r_r^{-1} = \nu^{\ast}$ occurs only if $\left(k_p/k_\omega\right)^2 = 2r_r^{-1}d + r_r^{-2}$, which is not likely to happen as mentioned above. Hence, we can conclude that $\|G_\text{iDroop}\|^{2\ast}_{H_2} < \|G_\text{DC}\|^2_{H_2}$ whenever $\left(k_p/k_\omega\right)^2 \neq 2r_r^{-1}d + r_r^{-2}$. 
\end{proof}
%It is possible that we are not that ambitious to obtain the optimal performance from iDroop control and what we want to do is just get a better performance from iDroop control than DC. If that is the case, after setting $\delta \to 0$, as long as $r_r^{-1} \neq \nu^{\ast}$, i.e. $\left(k_p/k_\omega\right)^2 \neq 2r_r^{-1}d + r_r^{-2}$, there must exist a range of value of $\nu$ that we can choose. To be precise, we simply need to compare $r_r^{-1}$ with $\nu^{\ast}$. If $r_r^{-1} < \nu^{\ast}$, then $\left(k_p/k_\omega\right)^2 > 2r_r^{-1}d + r_r^{-2}$, and then $\nu_1 < \nu_2$, so $\nu \in (r_r^{-1} , \nu^{\ast}) \cap(\nu_1 , \nu_2)=(r_r^{-1} , \nu^{\ast})$ is an acceptable region; if $r_r^{-1} > \nu^{\ast}$, then $\left(k_p/k_\omega\right)^2 < 2r_r^{-1}d + r_r^{-2}$, and then $\nu_1 > \nu_2$, so $\nu \in (\nu^{\ast} , r_r^{-1})\cap(\nu_2 , \nu_1)\cap[0,\infty)=(\nu^{\ast} , r_r^{-1})$ is an acceptable region.

Theorem \ref{th:h2-improves} shows that in order to optimally improve the performance iDroop needs to first set $\delta$ arbitrarily close to zero. Interestingly, this implies that the the transfer function $c_i(s)\simeq\nu$ except for $c_i(0)=r_r^{-1}$. In other words, iDroop uses its first order lead/lag property to effectively decouple the zero-frequency gain $c_i(0)$ from all the other frequencies $c_i(j\omega_0)=\simeq \nu$. This decouple is particularly easy to understand in two special regimes: (i) If $k_p\ll k_\omega$, the system is dominated by frequency measurements and therefore $\nu^{\ast} \simeq 0$. In this case,  we have $0\simeq\nu^*<r_r^{-1}$ which makes iDroop a lag compensator. Thus, by using lag compensation (setting $\nu<r_r^{-1}$) iDroop can attenuate frequency noise; 
(ii) If $k_p \gg k_\omega$, the system is dominated by power disturbances and $\nu^\ast \simeq k_p/k_\omega\gg 1$. Thus, for $k_p/k_\omega$ large enough, $\nu^*>r_r^{-1}$ and thus iDroop can use lead compensation (setting $\nu>r_r^{-1}$) to help mitigate power disturbances. 

%
%\begin{enumerate}%[indent=0pt]
%\item[(i)] If $k_p\ll k_\omega$, the system is dominated by frequency measurements and therefore $\nu^{\ast} \simeq 0$. That is, the main objective of iDroop is to attenuate frequency noise.
%\item[(ii)] If $k_p \gg k_\omega$, the system is dominated by power disturbances and $\nu^\ast \simeq k_p/k_\omega\gg 1$. Thus, iDroop can increase the gain to help mitigate power disturbances. 
%\end{enumerate}

%If $k_p \gg k_\omega$, then $\nu_2 \to \infty$, so $\nu \in (\nu_1, \infty)$ guarantees $\alpha_1 > 0$; if $k_\omega \gg k_p$, then $\nu_2 \simeq - r_r^{-1}( d) / (d + r_r^{-1}) < 0$, so $\nu \in [0,\nu_1)$ guarantees $\alpha_1 > 0$.

%%%%%%%%%%%%%%%%%%%%%%%%%%%%%%%%%%%%%%%%%%%%%%%%%%%%%%%%%%%%%%%%%%%%%%%%%%%%%%%%
% Delay Robustness
%%%%%%%%%%%%%%%%%%%%%%%%%%%%%%%%%%%%%%%%%%%%%%%%%%%%%%%%%%%%%%%%%%%%%%%%%%%%%%%%
%!TEX root = main.tex
\section{Delay Robustness}\label{sec:delay}

\subsection{Stability under Delays}

As discussed in Section~\ref{ssec:perfomrance-measures}, calculating the delay robustness amounts to computing the largest delays $\tau_i$ that maintain stability of the network model. In the heterogeneous inverter setting this can be addressed in a decentralized manner using Theorem 3.2 from \cite{pm2017ifac-wc}. To give simpler and more interpretable formulae that can be directly compared with those for $H_2$ performance, here we instead consider the homogeneous parameter setting. In this case stability criteria for any inverter design can be obtained directly from the multivariable Nyquist criterion of \cite{DW80}.

\begin{corollary}[\cite{DW80}]\label{cor:1}
Assuming homogeneous parameters and that $c\s$ is stable, the network model is stable if and only if
\[
\sum_{i=1}^{n}\text{w.n.o.}\frac{sc\s{}e^{-s\tau}}{ms^2+ds+\lambda{}_i}=0.
\]
In the above $\text{w.n.o.}$ denotes the winding number about the $-1$ point of the given transfer functions as evaluated on the usual Nyquist contour.
\end{corollary}
\begin{proof}
First observe that $P\s=\bar{P}\s\funof{I_n+C\bar{P}\s}^{-1}$, where $\bar{P}=\frac{1}{ms+d}I$ and $C\s=c\s{}I$. Therefore the network model is equivalent to $\bar{P}\s$ in negative feedback with $\frac{1}{s}L_B$ and $C\s$. By closing the network feedback loop, it then follows that the model is stable if and only if the negative feedback interconnection of $\bar{P}\s\funof{I_n+\frac{1}{s}L_B\bar{P}\s}^{-1}$ and $C\s$ is stable. Provided $d>0$, it is easily shown that $\bar{P}\s\funof{I_n+\frac{1}{s}L_B\bar{P}\s}^{-1}$ is stable (using e.g. a passivity argument), and so by the multivariable Nyquist criterion, the network model is stable if and only if
\[
\text{w.n.o.}\;\lambda\funof{C\s\bar{P}\s\funof{I_n+\frac{1}{s}L_B\bar{P}\s}^{-1}}=0.
\]
In the above $\text{w.n.o.}\;\lambda\funof{\cdot}$ denotes the winding number of the eigenloci of the given matrix, evaluated on the usual Nyquist contour; see \cite{DW80} for details. Since $C\s,\bar{P}\s$ and $L_B$ are all normal, commuting matrices, the eigenvalues of $C\s\bar{P}\s\funof{I_n+\frac{1}{s}L_B\bar{P}\s}^{-1}$ are easily shown to be
\[
\frac{sc\s{}e^{-s\tau}}{ms^2+ds+\lambda{}_i},
\]
from which the result follows immediately.
\end{proof}

\subsection{Delay Robustness of iDroop}

Corollary~\ref{cor:1} shows that in order to calculate the delay robustness, we need to calculate the largest $\tau$ such that a winding number condition is satsfied. This is straightforward to do numerically, however to facilitate comparison with $H_2$ performance analysis, we will explicitly compute the delay robustness when the iDroop controller is used in the two extreme regimes $\delta=0$ and $\delta\rightarrow{}\infty$.

\begin{theorem}[Delay Robustness of iDroop]\label{thm:delay}
Define
\[
\omega_n\funof{x}=\sqrt{\sqrt{x^2+\frac{2x\lambda_n}{m}}+x+\frac{\lambda_n}{m}}.
\]
Assuming homogeneous parameters, then:
\begin{enumerate}[(i)]
\item If $\delta=0$ and $d<\nu$, iDroop's delay robustness equals
\[
\tau_{\text{rob}}=\frac{1}{\omega_n\funof{\frac{1}{2m^2}\funof{\nu^2-d^2}}}\arccos\funof{-\frac{d}{\nu}}.
\]
\item If $\delta\rightarrow{}\infty$ and $d<r_r^{-1}$, then
\[
\tau_{\text{rob}}=\frac{1}{\omega_n\funof{\frac{1}{2m^2}\funof{r_r^{-2}-d^2}}}\arccos\funof{-\frac{d}{r_r^{-1}}}.
\]
\end{enumerate}
\end{theorem}
\begin{proof}
Observe that in both limiting cases,
\[
c\jw\rightarrow{}a,\forall{}\omega\in\R,
\]
where $a>0$ (for $\delta=0$, $a=\nu$, and for $\delta\rightarrow{}\infty{}$, $a=r_r^{-1}$). Therefore to prove the result, by Corollary~\ref{cor:1}, we need a method to calculate the winding numbers of the transfer functions
\begin{equation}\label{eq:simplenyq}
\frac{sae^{-s\tau}}{ms^2+ds+\lambda{}_i}.
\end{equation}
Observe that if $\tau\equiv{}0$, then the Nyquist diagrams for all these transfer functions are circles that cut the real axis at the origin and $\frac{a}{d}$. This implies that the phase margin $\phi$ for all these transfer functions are the same, and the winding number condition is satsified if and only if
\[
\max_{i\in{}V}\omega_i\tau<\phi.
\]
In the above, $\omega_i$ is the largest frequency at which the \emph{i}th transfer function has magnitude one (the frequency used to calculate the phase margin). A simple calculation shows that
\[
\cos\funof{\pi-\phi}=\frac{d}{a}.
\]
To complete the proof we then only need to find the frequencies $\omega_i$; that is solve the equation
\[
\abs{j\omega{}a}=\abs{-m\omega^2+dj\omega+\lambda{}_i}.
\]
This can be done by solving a fourth order polynomial equation. The general solution is given by
\[
\omega_i=\sqrt{\sqrt{x^2+\frac{2x\lambda_i}{m}}+x+\frac{\lambda_i}{m}},
\]
where $x=\frac{1}{2m^2}\funof{a^2-d^2}$. Since the above is monotonically inreasing in $\lambda_i$, the result follows.
\end{proof}

Although the above precisely gives $\tau_\text{rob}$ for iDroop in terms of the network parameters, it is a little hard to interpret. However if we assume that $d=0$, and replace $\omega_n\funof{x}$ with the upper bound
\[
\omega_n\funof{x}\leq{}\sqrt{2}\sqrt{x+\frac{m}{\lambda_n}},
\]
we get the following much simplified lower bound for $\tau_\text{rob}$:
\[
\tau_{\text{rob}}\geq{}\frac{m\pi}{2\sqrt{a^2+2m\lambda_n}}.
\]
In the above $a=\nu$ when $\delta=0$, and $a=r_r^{-1}$ when $\delta\rightarrow{}\infty$. We therefore clearly see that the delay robustness may be improved using iDroop by setting $\nu$ small, though this should be balanced with the desired objectives in the $H_2$ norm.
In particular, in systems where power disturbances dominates over frequency noise, jointly improving $H_2$ and delay robustness may be not possible.

\begin{remark}
In the case that $\nu\leq{}d$ or $r_r^{-1}\leq{}d$ (that is when we are outside the scope of Theorem~\ref{thm:delay}), the delay robusntess of iDroop is arbitrarily large. In such a regime, iDroop is stable when subject to delays of any size.
\end{remark}

%%%%%%%%%%%%%%%%%%%%%%%%%%%%%%%%%%%%%%%%%%%%%%%%%%%%%%%%%%%%%%%%%%%%%%%%%%%%%%%%
\section{Conclusions and Future Works}
This work aims to study the effect of using gird connected inverters to mitigate the dynamic degradation being experienced by the power grid. Our analysis shows that using dynamically controlled inverters like iDroop that drift away from the traditional virtual inertia and droop control mechanisms can provide a more efficient use of inverters that can further improve the power system performance in the presence of power disturbances and frequency noise, while providing larger robustness margins to delay. Future works include an extension of our analysis for heterogeneous system parameters as well as thorough experimental evaluation with detailed network models.

%\subsection{Conclusions}

%\subsection{Future Works}

\bibliographystyle{IEEEtran}
\bibliography{Refs}

%%%%%%%%%%%%%%%%%%%%%%%%%%%%%%%%%%%%%%%%%%%%%%%%%%%%%%%%%%%%%%%%%%%%%%%%%%%%%%%%
%\section{ACKNOWLEDGMENTS}

%%%%%%%%%%%%%%%%%%%%%%%%%%%%%%%%%%%%%%%%%%%%%%%%%%%%%%%%%%%%%%%%%%%%%%%%%%%%%%%%

\end{document}